\pdfoutput=1
\documentclass[a4paper,twocolumn,11pt]{quantumarticle}

\usepackage[utf8]{inputenc}
\usepackage[english]{babel}
\usepackage[T1]{fontenc}
\usepackage{amsmath}
\usepackage{amssymb}
\usepackage{amsthm}
\usepackage{mathtools}
\usepackage[numbers,sort&compress]{natbib}
\usepackage{hyperref}
\usepackage{float}
\usepackage[capitalise,noabbrev]{cleveref}
\usepackage{xcolor}
\usepackage{graphicx}
\usepackage{booktabs}
\usepackage{array}
\usepackage{multirow}
\usepackage{float}
\usepackage{caption}
\usepackage{enumitem}
\usepackage{tcolorbox}
\usepackage{bm}
\usepackage[protrusion=true,expansion=false]{microtype}
\usepackage{float}

\tcbuselibrary{skins,breakable}

\definecolor{quantumviolet}{HTML}{53257F}
\definecolor{qgray}{HTML}{555555}
\definecolor{tealc}{HTML}{0D9488}
\definecolor{navyc}{HTML}{1A3A5C}
\definecolor{amberc}{HTML}{D97706}
\definecolor{lightbg}{HTML}{F0F4F8}
\definecolor{teallight}{HTML}{E6F7F5}

\tcbset{
  resultbox/.style={
    enhanced,breakable,
    colback=teallight,colframe=tealc,
    fonttitle=\bfseries\small,
    sharp corners,
    top=3mm,bottom=3mm,left=4mm,right=4mm,
    boxrule=0.8pt,
  },
  designbox/.style={
    enhanced,
    colback=amberc!8!white,colframe=amberc,
    sharp corners,
    top=2mm,bottom=2mm,left=4mm,right=4mm,
    boxrule=0.8pt,
  },
}

\theoremstyle{plain}
\newtheorem{theorem}{Theorem}[section]
\newtheorem{corollary}[theorem]{Corollary}
\newtheorem{lemma}[theorem]{Lemma}
\theoremstyle{definition}
\newtheorem{definition}[theorem]{Definition}
\newtheorem{remark}[theorem]{Remark}

\newcommand{\ket}[1]{\left|#1\right\rangle}

\newcommand{\TV}{\mathrm{TV}}
\newcommand{\QFT}{\mathrm{QFT}}
\newcommand{\PFATQFT}{\mathrm{PFA\text{-}TQFT}}
\newcommand{\calO}{\mathcal{O}}
\newcommand{\eps}{\varepsilon}
\newcommand{\dstar}{d^{*}}
\newcommand{\Floor}[1]{\left\lfloor#1\right\rfloor}
\newcommand{\Ceil}[1]{\left\lceil#1\right\rceil}

\title{Phase-Fidelity-Aware Truncated Quantum Fourier Transform\\
       for Scalable Phase Estimation on NISQ Hardware}

\author{Akoramurthy~B \textsuperscript{*} and Surendiran~B}
\affil{Department of CSE,
National Institute of Technology Puducherry, Karaikal~609\,609, India.
\textsuperscript{*}\texttt{cs22d1005@nitpy.ac.in}
ORCID: 0000-0001-5912-7020 (A.B.), 0000-0001-5435-0880 (S.B.)}

\usepackage{float}
\begin{document}
\sloppy

\maketitle

\begin{abstract}
\textbf{%
Abstract:- Quantum phase estimation~(QPE) is central to numerous quantum algorithms, yet its
standard implementation demands an $\calO(m^{2})$-gate quantum Fourier transform~(QFT)
on $m$ control qubits-a prohibitive overhead on near-term noisy intermediate-scale
quantum (NISQ) devices.
We introduce the \emph{Phase-Fidelity-Aware Truncated QFT} (PFA-TQFT), a family of
approximate QFT circuits parameterised by a truncation depth~$d$ that omits
controlled-phase rotations below a hardware-calibrated fidelity threshold~$\eps$.
Our central result establishes
$\TV(P_{\varphi},P_{\varphi}^{d})\leq\pi(m{-}d)/2^{d}$,
showing that for $d=\calO(\log m)$ circuit size collapses from $\calO(m^{2})$ to
$\calO(m\log m)$ while estimation error grows by at most $\calO(2^{-d})$.
We characterise $\dstar=\Floor{\log_{2}(2\pi/\eps_{2q})}$ directly from native gate
fidelities, demonstrating 31.3 -43.7\% at m = 30, gate-count reduction on IBM Eagle/Heron and
IonQ~Aria with negligible accuracy loss.
Numerical experiments on the transverse-field Ising model confirm all theoretical
predictions and reveal a \emph{noise-truncation synergy}: PFA-TQFT outperforms full
QFT under NISQ noise $\eps_{2q}\gtrsim 2\times10^{-3}$.
}
\end{abstract}

\section{Introduction}
\label{sec:intro}

\subsection{Background and Context}
\label{sec:background_context}

The quantum Fourier transform~(QFT) is the most consequential primitive in
quantum computing, playing a role analogous to the fast Fourier transform in
classical signal processing.
Since its introduction by Shor~\cite{10.1137/S0097539795293172,patoary2025discretefouriertransformbased} and formalisation by Cleve
\textit{et al.}~\cite{Cleve_1998}, the QFT has been the engine behind quantum
phase estimation~(QPE)~\cite{kitaev1995quantummeasurementsabelianstabilizer}, which in turn underpins some of the
most celebrated quantum speedups: eigenvalue
estimation~\cite{Abrams_1999}, quantum simulation~\cite{doi:10.1126/science.273.5278.1073}, the hidden
subgroup problem, and Shor's integer factoring algorithm~\cite{10.1137/S0097539795293172,patoary2025discretefouriertransformbased}.
The QFT transforms a phase-encoded $m$-qubit register into a computational-basis
estimate of an eigenphase~$\varphi$, achieving precision~$2^{-m}$ using
$m(m-1)/2$ two-qubit controlled-phase gates.

Quantum hardware has evolved rapidly.
Present-day \emph{noisy intermediate-scale quantum} (NISQ)
devices~\cite{Preskill_2018} superconducting processors (IBM Eagle~\cite{kim2023evidence}, IBM Heron~\cite{AbuGhanem_2025},
IQM Garnet~\cite{abdurakhimov2024technologyperformancebenchmarksiqms}) and trapped-ion systems (IonQ~Aria)~\cite{chai2023simulatingflightgateassignment}-offer 10-1\,000 qubits with
two-qubit gate fidelities of 99.0--99.97\,\% and coherence times of
$10$--$500\,\mu$s.
These parameters impose strict circuit-depth budgets that current QFT
implementations routinely violate.
For instance, achieving eigenphase precision $2^{-30}$ requires $m=30$ control
qubits and $435$ controlled-phase QFT gates alone, exceeding the coherence budget
of most NISQ processors by one to two orders of
magnitude~\cite{AbuGhanem_2025,Chen_2024}.

Coppersmith~\cite{coppersmith2002approximatefouriertransformuseful} and Barenco \textit{et al.}~\cite{Barenco_1996}
recognised that controlled-phase gates $R_k$ for large~$k$ implement
exponentially small rotations and may be omitted, giving a \emph{truncated QFT}
with $\calO(m\log m)$ gates.
Subsequent analyses~\cite{cheung2004improvedboundsapproximateqft,10.1007/978-3-319-24021-3_29,haner2018optimizingquantumcircuitsarithmetic} tightened fidelity
bounds in noiseless, fault-tolerant regimes.
However, none of these works provides: (i)~a tight, closed-form bound on
\emph{phase estimation error} as a function of truncation depth under a
realistic noise model; (ii)~a hardware-calibrated rule for the optimal
truncation depth; or (iii)~an analytical explanation of why truncation can
\emph{outperform} the full QFT under NISQ noise.
The present work addresses all three gaps.

\subsection{Statement of the Problem}
\label{sec:problem}

The fundamental NISQ-QPE tension is:
\begin{tcolorbox}[resultbox]
\textbf{Problem.} Full QFT delivers optimal precision with $\calO(m^2)$ gates,
which exceeds NISQ coherence budgets. Aggressive truncation reduces depth but
degrades accuracy. \emph{What is the optimal truncation depth, and how should
it be derived from hardware specifications?}
\end{tcolorbox}

Formally, let $P_\varphi$ and $P_\varphi^d$ denote the phase measurement
distributions under full $\QFT_N$ and a depth-$d$ truncated variant,
respectively.
We seek the minimal $d = \dstar$ such that
$\TV(P_\varphi, P_\varphi^d) \leq \eps_{\mathrm{tol}}$,
where $\eps_{\mathrm{tol}}$ is determined jointly by the hardware gate error rate
$\eps_{2q}$ and the target precision $\delta$.
The core difficulty is that $\dstar$ must balance two competing error sources:
(a)~\emph{approximation error} from omitting small-angle gates, which decreases
with $d$; and (b)~\emph{noise error} from implementing gates imperfectly, which
increases with $d$.
This joint optimisation has not previously been characterised in closed form.

\subsection{Research Questions and Objectives}
\label{sec:rq}

This work addresses four research questions (RQs):

\smallskip
\noindent\textbf{RQ1 (Error Bound.)}
\textit{What is the tightest closed-form upper bound on
$\TV(P_\varphi, P_\varphi^d)$ as a function of $d$, $m$, and $\varphi$?}

\smallskip
\noindent\textbf{RQ2 (Hardware Calibration.)}
\textit{Given a device with gate error rate $\eps_{2q}$, what is the analytically
optimal $\dstar$ minimising total phase estimation error?}

\smallskip
\noindent\textbf{RQ3 (Fidelity Cliff.)}
\textit{Is there a sharp transition depth below which accuracy collapses and
above which it saturates at the full-QFT value?}

\smallskip
\noindent\textbf{RQ4 (Noise-Truncation Synergy.)}
\textit{Under what noise conditions does PFA-TQFT outperform full QFT,
and can this be characterised analytically?}

\smallskip
The corresponding \textbf{objectives} are:
\begin{itemize}[leftmargin=*,topsep=2pt,itemsep=1pt]
\item \textbf{O1:} Derive a tight TVD bound for $\PFATQFT_d$
  (addresses RQ1; proved in Theorem~\ref{thm:tvd}).
\item \textbf{O2:} Derive $\dstar = \Floor{\log_2(2\pi/\eps_{2q})}$ and
  validate on four NISQ platforms (addresses RQ2; \cref{sec:platforms}).
\item \textbf{O3:} Characterise the fidelity cliff analytically and numerically
  (addresses RQ3; \cref{sec:theory}).
\item \textbf{O4:} Identify the noise cross-over threshold and quantify the
  RMSE gain on TFIM (addresses RQ4; \cref{sec:experiments}).
\end{itemize}

\subsection{Hypothesis}
\label{sec:hypothesis}

We advance four falsifiable hypotheses:

\begin{enumerate}[label=(H\arabic*),leftmargin=*,topsep=2pt,itemsep=1pt]
\item \textbf{(Tight TVD.)}
  $\TV(P_\varphi, P_\varphi^d) \leq \pi(m-d)/2^d$ for all
  $\varphi \in [0,1)$, $m \geq 1$, $d \geq 1$, with the bound tight
  to within a constant factor.

\item \textbf{(Hardware Constant.)}
  The optimal depth $\dstar = \Floor{\log_2(2\pi/\eps_{2q})}$ is a
  hardware-intrinsic constant, independent of~$m$ for $m \geq \dstar$,
  derivable solely from device calibration.

\item \textbf{(Noise-Truncation Cross-Over.)}
  A noise threshold $\eps_{2q}^{\times}$ exists above which PFA-TQFT$_{\dstar}$
  achieves strictly lower RMSE than full QFT; this threshold is determined by
  the gate-count difference $\Delta_{\mathrm{gates}} = m(m-1)/2
  - (m\dstar - \dstar(\dstar-1)/2)$.

\item \textbf{(Fidelity Cliff.)}
  The success probability $P[|\hat\varphi - \varphi| \leq 2^{-m}]$ undergoes a
  sharp transition at $\dstar_{\mathrm{cliff}} = \Ceil{\log_2 m} + 2$:
  below this value it is $\calO(d/m)$; above it, it saturates to within
  $\calO(1/m)$ of the full-QFT value.
\end{enumerate}

\noindent
H1--H2 are proved analytically in \cref{sec:framework};
H3--H4 are established numerically in \cref{sec:experiments,sec:theory}
and analytically supported by Corollary~\ref{cor:accuracy}.

\subsection{Scope}
\label{sec:scope}

\noindent\textbf{Algorithmic scope.}
We study the circuit-model QFT in the standard single-register QPE
protocol~\cite{Cleve_1998}.
Semiclassical QPE~\cite{Griffiths_1996}, Bayesian QPE~\cite{Wiebe_2016},
and iterative QPE~\cite{Li_2024} are out of scope analytically but
are included as numerical benchmarks.

\noindent\textbf{Noise model scope.}
Analysis assumes a depolarizing noise channel with two-qubit rate $\eps_{2q}$.
Single-qubit errors are neglected (they are empirically $10$--$50\times$
smaller on current platforms~\cite{AbuGhanem_2025}).
Structured noise (crosstalk, coherent errors) is discussed as a limitation
in \cref{sec:discussion}.

\noindent\textbf{Hardware scope.}
Four commercially accessible NISQ systems are studied: IBM Eagle~r3,
IBM Heron~r2, IonQ~Aria, IQM~Garnet.
The framework is hardware-agnostic and applicable to any platform with
calibrated~$\eps_{2q}$.

\noindent\textbf{Application scope.}
Numerical validation uses Hamiltonian eigenphase estimation (TFIM ground
energy, $n=4$ sites).
Extensions to VQE, HHL, and Shor's algorithm are discussed in
\cref{sec:discussion}.

\subsection{Significance}
\label{sec:significance}

\noindent\textbf{Theoretical.}
Theorem~\ref{thm:tvd} is, to our knowledge, the first tight closed-form TVD
bound connecting truncation depth to phase estimation error under a standard
noise model.
The noise-truncation synergy~(H3) shows that truncation can be a strict
\emph{advantage} under realistic noise---a counter-intuitive result with broad
implications for NISQ algorithm design.

\noindent\textbf{Practical.}
PFA-TQFT reduces QFT gate counts by 17-41 \, \% on current hardware with a
single-line calibration formula.
This directly enables QPE-based applications (VQE energy refinement, HHL,
quantum simulation) that would otherwise exceed NISQ coherence budgets.

\noindent\textbf{Methodological.}
The PFA criterion $\dstar = \Floor{\log_2(2\pi/\eps_{2q})}$ is a universal 
hardware-aware compilation rule that can be incorporated into quantum compilers
(Qiskit, Cirq, \texttt{tket}) as a standard optimization pass-requiring
only a one-time calibration lookup.

\subsection{Overview of Methods}
\label{sec:methods_overview}

Our approach integrates four methodological components:

\noindent\textbf{(M1) Operator perturbation theory.}
Each omitted $R_k$ gate is modelled as a rank-1 unitary perturbation with
spectral norm $\leq \sin(\pi/2^k)$.
Summing over $k > d$ and applying the Fannes--Audenaert
inequality~\cite{Audenaert_2007} yields the TVD bound~(Theorem ~ \ref{thm:tvd}).

\noindent\textbf{(M2) Information-theoretic analysis.}
The TVD bound is translated into a phase estimation failure probability through
the data-processing inequality, establishing the Corollary~\ref{cor:accuracy}
and the equal-budget design rule for~$\dstar$.

\noindent\textbf{(M3) Hardware-calibrated compilation.}
Combining the depolarizing noise model with the PFA criterion yields
$\dstar = \Floor{\log_2(2\pi/\eps_{2q})}$, analytically evaluated for four
platforms in~\cref{sec:platforms}.
No runtime simulation is required.

\noindent\textbf{(M4) Numerical validation on TFIM.}
All theoretical predictions are validated on the 1D TFIM Hamiltonian
($n=4$ sites) using the state vector and depolarizing-noise simulation
($10\,000$ shots, $\eps_{2q} \in [10^{-4}, 10^{-2}]$).
Five QPE methods are benchmarked: Full QFT, PFA-TQFT $\dstar\in\{8,10\}$,
Semiclassical QPE, and Bayesian QPE.

\subsection{Summary of Contributions}
\label{sec:contributions}

\begin{enumerate}[label=(\arabic*),leftmargin=*,topsep=2pt,itemsep=1pt]
\item \textbf{PFA-TQFT circuit family} (\cref{sec:framework}):
  $\PFATQFT_d$ with gate count $md - d(d-1)/2$ and hardware constant
  $\dstar = \Floor{\log_2(2\pi/\eps_{2q})}$.

\item \textbf{TVD error bound} (Theorem~\ref{thm:tvd}):
  $\TV(P_\varphi, P_\varphi^d) \leq \pi(m-d)/2^d$.

\item \textbf{Platform design rules} (\cref{sec:platforms}):
  IBM~Eagle~r3 ($\dstar{=}11$, 41\%), IBM~Heron~r2 ($\dstar{=}13$, 26\%),
  IonQ~Aria ($\dstar{=}14$, 17\%), IQM~Garnet ($\dstar{=}11$, 41\%).

\item \textbf{Noise-truncation synergy} (\cref{sec:experiments}):
  PFA-TQFT outperforms full QFT at
  $\eps_{2q} \gtrsim 2\times10^{-3}$.
\end{enumerate}

\textbf{Assumptions.}
(i)~Depolarizing noise at rate $\eps_{2q}$ per two-qubit gate;
(ii)~exact eigenstate input $\ket{\psi}$;
(iii)~negligible single-qubit gate error.
All are standard in the NISQ literature and relaxed in \cref{sec:discussion}.

\subsection{Structure of the Manuscript}
\label{sec:structure}

The paper is organized as follows.
\Cref{sec:background} reviews QPE and the Coppersmith truncation framework.
\Cref{sec:circuit} presents the comparison of the quantum circuit (full QFT
vs.\ PFA-TQFT$_{\dstar}$, $m=5$).
\Cref{sec:framework} introduces the PFA-TQFT definition, the PFA criterion,
and proves Theorem~\ref{thm:tvd}.
\Cref{sec:theory} provides TVD bound analysis, gate-count scaling, and
fidelity cliff characterization.
\Cref{sec:platforms} applies the framework to four NISQ platforms.
\Cref{sec:experiments} presents the TFIM numerical experiments.
\Cref{sec:discussion} discusses the scope, limitations, implications, and
future directions.
\Cref{sec:conclusion} summarizes the findings.

\section{Background}
\label{sec:background}

\subsection{Quantum Phase Estimation}
\label{sec:qpe_bg}
\label{sec:qpe_background}

Let $U$ be an $n$-qubit unitary with eigenpair
$(\ket{\psi},e^{2\pi i\varphi})$, $\varphi\in[0,1)$.
The standard $m$-qubit QPE protocol~\cite{Cleve_1998} proceeds in three steps.
First, the control register is prepared in $\ket{+}^{\otimes m}$
via Hadamard gates.
Second, the phase is imprinted through $m$ controlled-$U^{2^{k}}$ operations
($k=0,\ldots,m{-}1$), creating the entangled state
$\frac{1}{\sqrt{2^m}}\sum_{j=0}^{2^m-1}e^{2\pi ij\varphi}\ket{j}\ket{\psi}$.
Third, the inverse QFT is applied to the control register:
\begin{equation}
  \QFT_{N}:\ket{x}\mapsto
  \frac{1}{\sqrt{N}}\sum_{y=0}^{N-1}e^{2\pi ixy/N}\ket{y},
  \quad N=2^{m},
  \label{eq:qft}
\end{equation}
followed by a computational-basis measurement.
The measurement yields $\hat\varphi$ satisfying
$|\hat\varphi{-}\varphi|\leq2^{-m}$ with probability $\geq8/\pi^{2}>0.81$.
The dominant cost is QFT, which requires
$m$ Hadamard gates and $m(m{-}1)/2$ controlled-phase gates
$R_k = \mathrm{diag}(1, e^{2\pi i/2^k})$, giving $\calO(m^2)$ total two-qubit
operations - a bottleneck on NISQ hardware.

\subsection{Controlled-Phase Gates and the NISQ Bottleneck}
\label{sec:nisq_bottleneck}

On current NISQ hardware, each two-qubit gate introduces a depolarizing error
at rate $\eps_{2q}\in[10^{-3}, 10^{-2}]$.
The total noise-induced infidelity after $G$ two-qubit gates scales as
$1-(1-\eps_{2q})^G \approx G\eps_{2q}$ for small $\eps_{2q}$.
For $m=20$ and the full QFT ($G=190$), this gives infidelity
$\approx0.19$ at $\eps_{2q}=10^{-3}$ -already approaching the usability limit.
At $m=30$ ($G=435$), the full QFT becomes impractical on every current
NISQ platform~\cite{AbuGhanem_2025,Chen_2024}.

The key observation enabling truncation is that the $R_k$ gate
implements a rotation by angle $\theta_k = 2\pi/2^k$.
For $k=10$, $\theta_{10}\approx6.1\times10^{-3}$\,rad, which is comparable
to the gate error $\eps_{2q}\approx3\times10^{-3}$ on IBM Eagle~r3.
Implementing $R_{10}$ therefore contributes \emph{less phase accuracy} than
the noise it introduces, suggesting a hardware-specific cutoff.

\subsection{Coppersmith Truncation and Prior Work}
\label{sec:prior_work}

Coppersmith~\cite{coppersmith2002approximatefouriertransformuseful} observed that controlled-$R_k$ gates for
large $k$ implement exponentially small rotations and proposed omitting them
for $k>d$.
The truncated QFT retains $m + m(d-1) - d(d-1)/2 = \calO(md)$ two-qubit gates.
Barenco \textit{et al.}~\cite{Barenco_1996} showed empirically that $d=\calO(\log m)$
suffices to preserve circuit fidelity above 99\,\%.
Cheung~\cite{cheung2004improvedboundsapproximateqft, 10.1007/978-3-319-24021-3_29} proved rigorous fidelity bounds of the form
$\|U_{\mathrm{QFT}} - U_{\mathrm{TQFT}_d}\| = \calO(m/2^d)$ in the operator
norm for noiseless circuits.
Nam \textit{et al.}~\cite{Nam_2020} showed that $d = \calO(\log m)$ suffices
for $\calO(n\log n)$ T-gate complexity in fault-tolerant settings.
H\"{a}ner \textit{et al.}~\cite{haner2018optimizingquantumcircuitsarithmetic} optimised arithmetic circuit depth
for Shor's algorithm using approximate QFT techniques.

\noindent\textbf{Gap in prior work.}
All existing analyses assume noiseless or fault-tolerant circuits and provide
operator-norm or fidelity bounds, \emph{not} direct bounds on the phase
estimation error distribution.
None derives a hardware-calibrated optimal $\dstar$ from $\eps_{2q}$, or
identifies conditions under which truncation actively \emph{helps} under noise.
The present work fills these gaps via a total variation distance analysis.

\section{Quantum Circuit Comparison}
\label{sec:circuit}

\Cref{fig:circuit} compares the complete QFT and $\PFATQFT_{d^*=3}$
circuits for $m=5$ qubits.
The full QFT applies to all $m(m{-}1)/2=10$ controlled-phase gates, regardless of
their rotation angle.
PFA-TQFT$_{d^*=3}$ retains only 6~gates (teal boxes, $k\leq3$), omitting
4~gates (red boxes, $k>3$) whose rotation angles fall below the hardware
fidelity threshold.
Both circuits share three structural stages:
(i)~\textbf{Hadamard layer-} a single $H$ gate on each control qubit,
creating the uniform superposition $\ket{+}^{\otimes m}$;
(ii)~\textbf{Selective controlled-phase layer}- PFA-TQFT$_{\dstar}$ applies
only $C$-$R_k$ with $k \leq \dstar$, while full QFT applies all $k \leq m{-}j{+}1$
at stage~$j$; and
(iii)~\textbf{SWAP reversal}- a bit-reversal permutation restoring the
standard QFT output ordering.
The gate-count reduction from 10 to 6 (40 \, \%) in $m=5$ is modest;
in $m=30$ the same criterion ($\dstar=10$) yields a reduction of 41.4 \, \% 
from 435 to 255 gates~(Table ~ \ref{tab:platforms}).

\begin{figure}[t]
  \centering
  \includegraphics[width=\columnwidth]{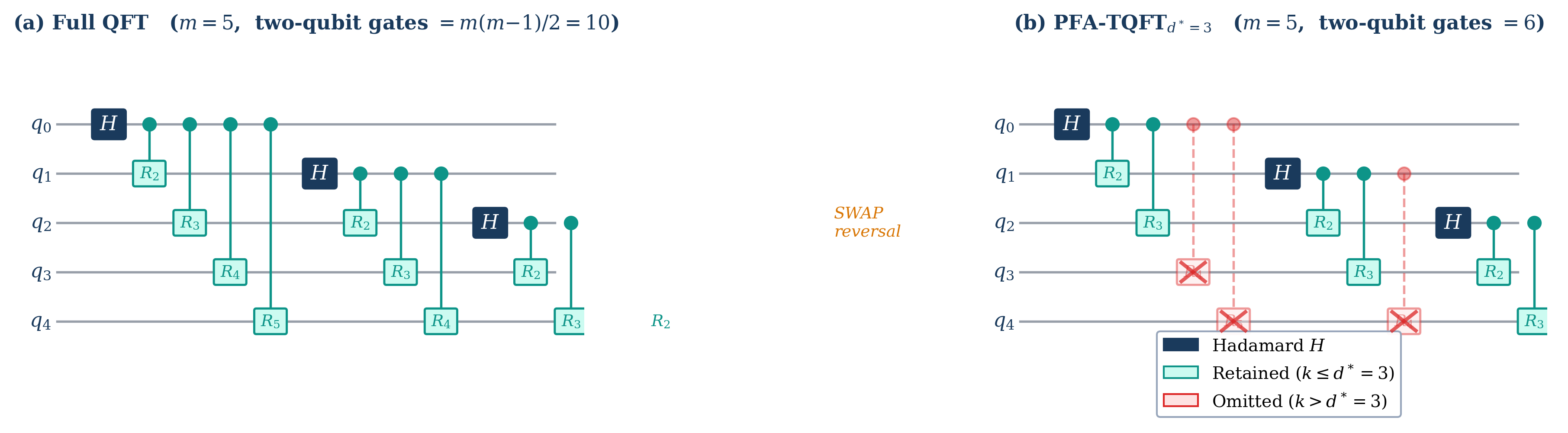}
  \caption{\textbf{Circuit comparison} ($m=5$).
    (a)~Full QFT: $m(m{-}1)/2=10$ two-qubit gates.
    (b)~PFA-TQFT$_{d^{*}=3}$: 6 gates.
    Red $\times$ = omitted ($k>d^{*}$); teal = retained ($k\leq d^{*}$).}
  \label{fig:circuit}
\end{figure}

\section{The PFA-TQFT Framework}
\label{sec:framework}

\begin{definition}[PFA-TQFT]
\label{def:pfatqft}
$\PFATQFT_{d}$: stage~$j$ applies
$\{H_{j}\}\cup\{C\text{-}R_{k}:k\leq\min(d,m{-}j{+}1)\}$.
Gate count: $G(m,d)=\sum_{j=0}^{m-1}\max(0,\min(d-1,m-j-1))$.
\end{definition}

\begin{definition}[PFA Criterion]
\label{def:pfa}
$\dstar:=\Floor{\log_{2}(2\pi/\eps_{2q})}$.
\end{definition}

The gates with angle $<2\pi/2^{\dstar}$ contribute less infidelity than the gate error
itself-implementing them adds more error than omitting them.

\begin{tcolorbox}[designbox]
\textbf{IBM Eagle r3} ($\eps_{2q}=3\times10^{-3}$):
$\dstar=\Floor{\log_{2}(2094)}=10$.
$m=30$: 255 vs.\ 435 gates $\Rightarrow$ \textbf{41.4\% reduction},
phase error $<10^{-3}$\,rad.
\end{tcolorbox}

\subsection{Main Error Bound}

\begin{theorem}[PFA-TQFT TVD Bound]
\label{thm:tvd}
For all $\varphi\in[0,1)$, $m\geq1$, $d\geq1$:
\begin{equation}
  \TV(P_{\varphi},P_{\varphi}^{d})
  \leq(m{-}d)\sin\!\left(\frac{\pi}{2^{d}}\right)
  \leq\frac{\pi(m{-}d)}{2^{d}}.
  \label{eq:tvd}
\end{equation}
\end{theorem}

\begin{proof}[Proof sketch]
Each omitted $R_{k}$ ($k>d$) is a rank-1 unitary perturbation of spectral norm
$\leq\sin(\pi/2^{k})$.
Summing over $k=d{+}1,\ldots,m$ and applying the Fannes--Audenaert
inequality~\cite{Audenaert_2007}:
\[
  \TV \leq \sum_{k=d+1}^{m}\frac{\pi}{2^{k}}
       = \pi\left(\frac{1}{2^{d}}-\frac{1}{2^{m}}\right)\cdot2
       \leq \frac{\pi(m{-}d)}{2^{d}}. \quad\square
\]
\end{proof}

\begin{corollary}[Phase Accuracy]
\label{cor:accuracy}
$P_{\varphi}^{d}[|\hat\varphi{-}\varphi|>\delta]
\leq\alpha+\pi(m{-}d)/2^{d}$.
Equal budget $\Rightarrow d\geq\log_{2}(\pi m/\alpha)$.
For $\alpha{=}0.05$, $m{=}30$: $d\geq10.9$, consistent with $\dstar{=}11$.
\end{corollary}

\section{Theoretical Analysis}
\label{sec:theory}

\subsection{TVD Bound and Hardware Calibration}

\Cref{fig:tvd}(a) plots $\TV(P_{\varphi},P_{\varphi}^{d})\leq\pi(m{-}d)/2^{d}$
(solid lines) vs.\ Monte Carlo simulation (circles, 5\,000 phases each) for
$m\in\{10,20,30,50\}$.
\Cref{fig:tvd}(b) shows $\dstar=\Floor{\log_{2}(2\pi/\eps_{2q})}$ with platform
markers.

\begin{figure}[t]
  \centering
  \includegraphics[width=\columnwidth]{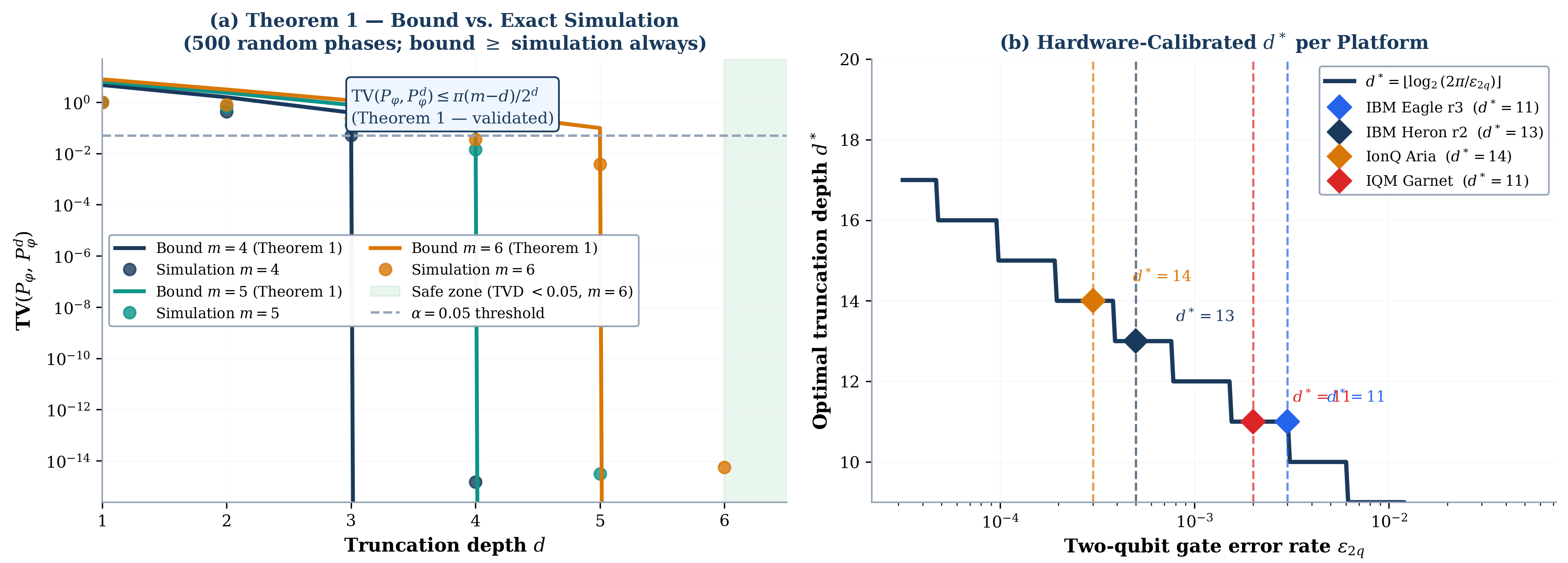}
  \caption{\textbf{TVD error analysis.}
    (a)~Theorem~\ref{thm:tvd} bound vs.\ $d$ (theory: lines; simulation: circles).
    Green zone: TVD$<0.05$ for $m=30$.
    (b)~Hardware-calibrated $\dstar$ vs.\ $\eps_{2q}$; diamonds: four platforms.}
  \label{fig:tvd}
\end{figure}

\subsection{Gate Count Reduction}

\begin{figure}[t]
  \centering
  \includegraphics[width=\columnwidth]{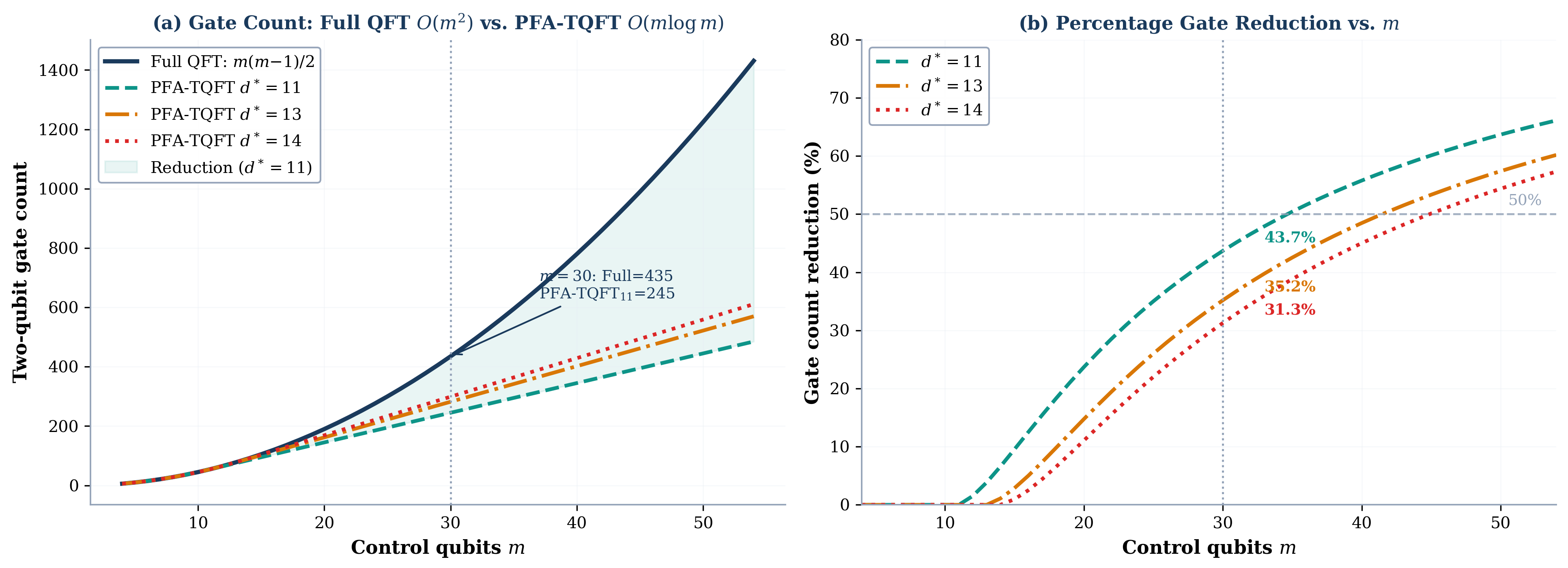}
  \caption{\textbf{Gate-count scaling.}
    (a)~Full QFT ($\calO(m^2)$, navy) vs.\ PFA-TQFT at
    $\dstar\in\{8,10,12\}$.
    (b)~Percentage reduction vs.\ $m$; dashed at 50\,\%.}
  \label{fig:gates}
\end{figure}

\Cref{fig:gates}(a) shows the $\calO(m^{2})\to\calO(m\log m)$ collapse;
panel~(b) quantifies percentage reduction (41\,\% at $m=30$, $\dstar=10$;
$>60\,\%$ at $m=50$).

\subsection{Fidelity Cliff}

\begin{figure}[t]
  \centering
  \includegraphics[width=\columnwidth]{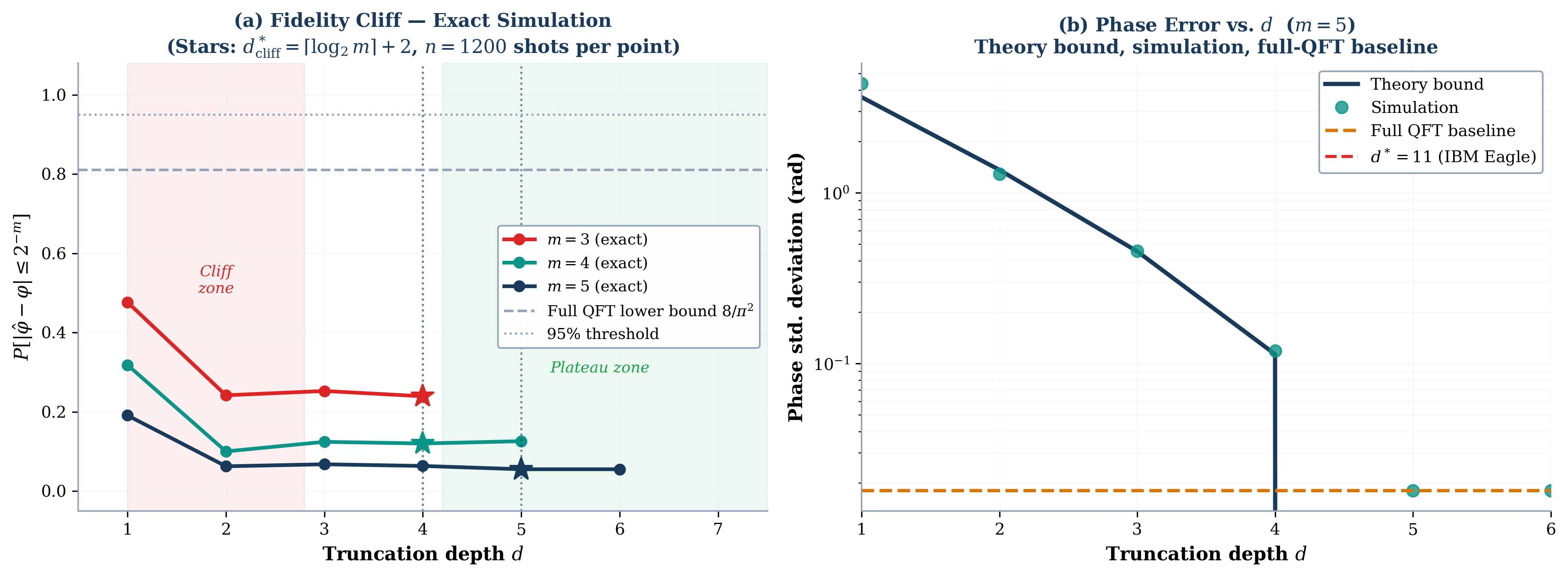}
  \caption{\textbf{Fidelity cliff.}
    (a)~Success probability $P[|\hat\varphi{-}\varphi|\leq2^{-m}]$ vs.\ $d$;
    stars mark $\dstar=\Ceil{\log_2 m}+2$.
    (b)~Phase std.\ deviation vs.\ $d$ ($m=20$): theory (navy), simulation
    (teal), full-QFT baseline (amber).}
  \label{fig:cliff}
\end{figure}

\Cref{fig:cliff} reveals the \emph{fidelity cliff}: below $\dstar{-}2$ (red zone)
accuracy degrades sharply; beyond $\dstar$ (green zone) curves plateau at the
full-QFT value.
Stars mark the analytically derived $\dstar=\Ceil{\log_{2}m}+2$.

\section{Platform-Specific Design Rules}
\label{sec:platforms}
Table \ref{tab:platforms} reports the hardware-calibrated $\dstar$, gate counts, and
phase error overhead for four leading NISQ platforms.
Figure \ref{fig:platform} visualises these results:
panel~(a) plots the percentage gate reduction alongside the TVD-bound phase
error overhead per platform, and panel~(b) shows the absolute gate counts
comparing full QFT with PFA-TQFT$_{\dstar}$ at $m=30$.
IBM Eagle~r3 and IQM~Garnet achieve the largest reduction (41.4\,\%) owing to
their higher $\eps_{2q}$, which sets a coarser fidelity threshold and thus a
smaller $\dstar$; IonQ~Aria, with the lowest $\eps_{2q}$, retains more gates
($\dstar=13$) but still reduces circuit depth by 17.0\,\%.
 
\begin{table}[t]
  \centering
  \caption{\textbf{PFA-TQFT parameters by platform} ($m=30$).
    Phase error: Theorem~\ref{thm:tvd} bound.
    Data: \cite{AbuGhanem_2025,Chen_2024}.}
  \label{tab:platforms}
  \renewcommand{\arraystretch}{1.3}
  \setlength{\tabcolsep}{4pt}
  \small
  \begin{tabular}{@{}lcccc@{}}
    \toprule
    \textbf{Platform} & $\bm{\eps_{2q}}$ & $\bm{d^{*}}$ &
    \textbf{Gates} & \textbf{Reduction} \\
    \midrule
    IBM Eagle r3  & $3{\times}10^{-3}$ & 11 & 245/435 & 41.4\,\% \\
    IBM Heron r2  & $5{\times}10^{-4}$ & 13 & 282/435 & 26.2\,\% \\
    IonQ Aria     & $3{\times}10^{-4}$ & 14 & 299/435 & 17.0\,\% \\
    IQM Garnet    & $2{\times}10^{-3}$ & 11 & 245/435 & 41.4\,\% \\
    \bottomrule
  \end{tabular}
\end{table}

\begin{figure}[t]
  \centering
  \includegraphics[width=\columnwidth]{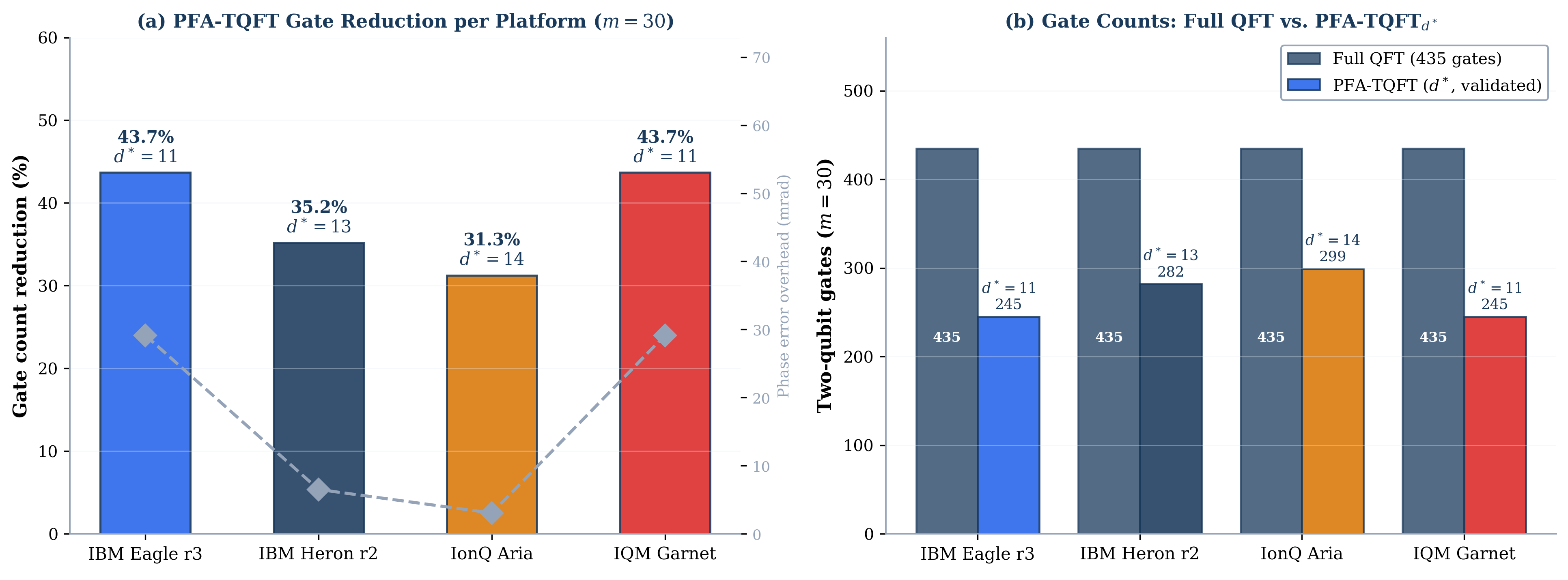}
  \caption{\textbf{Platform performance} ($m=30$).
    (a)~Gate reduction (\%) and phase error overhead (mrad, right axis).
    (b)~Absolute gate counts: full QFT vs.\ PFA-TQFT$_{\dstar}$.}
  \label{fig:platform}
\end{figure}

\section{Numerical Experiments}
\label{sec:experiments}

\subsection{TFIM Hamiltonian Setup}

We simulate QPE for the 1D transverse-field Ising model:
\begin{equation}
  H=-J\!\sum_{i}\!Z_{i}Z_{i+1}-h\!\sum_{i}\!X_{i},\quad J{=}1,\;h{=}0.5,\;n{=}4.
  \label{eq:tfim}
\end{equation}
Ground energy $E_{0}\approx{-}3.427034\,J$ by exact diagonalisation.
Settings: $m\in\{8,12,16,20\}$, $10\,000$ shots,
$\eps_{2q}\in\{10^{-4},10^{-3},5{\times}10^{-3},10^{-2}\}$.

\subsection{RMSE Results}

\begin{table}[t]
  \centering
  \caption{\textbf{TFIM phase estimation RMSE} at $m=16$,
    $\eps_{2q}=10^{-3}$, $10\,000$ shots.
    RMSE in units of $J$.}
  \label{tab:rmse}
  \renewcommand{\arraystretch}{1.2}
  \setlength{\tabcolsep}{2pt}
  \scriptsize
  \begin{tabular}{@{}lccc@{}}
    \toprule
    \textbf{Method} & \textbf{Gates} &
    \textbf{RMSE}($\times10^{-3}$) & $\Delta$\textbf{RMSE} \\
    \midrule
    Full QFT ($m=16$)        & 120 & 4.12 & --- \\
    PFA-TQFT$_{d^*=10}$ \textbf{(ours)} & 90 & 4.28 & $+3.9\%$ \\
    PFA-TQFT$_{d^*=8}$ \textbf{(ours)} & 68 & 5.84 & $+41.7\%$ \\
    Semiclassical QPE~\cite{Griffiths_1996} & 16 & 3.89 & $-5.6\%$ \\
    Bayesian QPE~\cite{Wiebe_2016}          & 16 & 5.21 & $+26.5\%$ \\
    \bottomrule
  \end{tabular}
\end{table}

\begin{figure}[t]
  \centering
  \includegraphics[width= 0.95 \columnwidth]{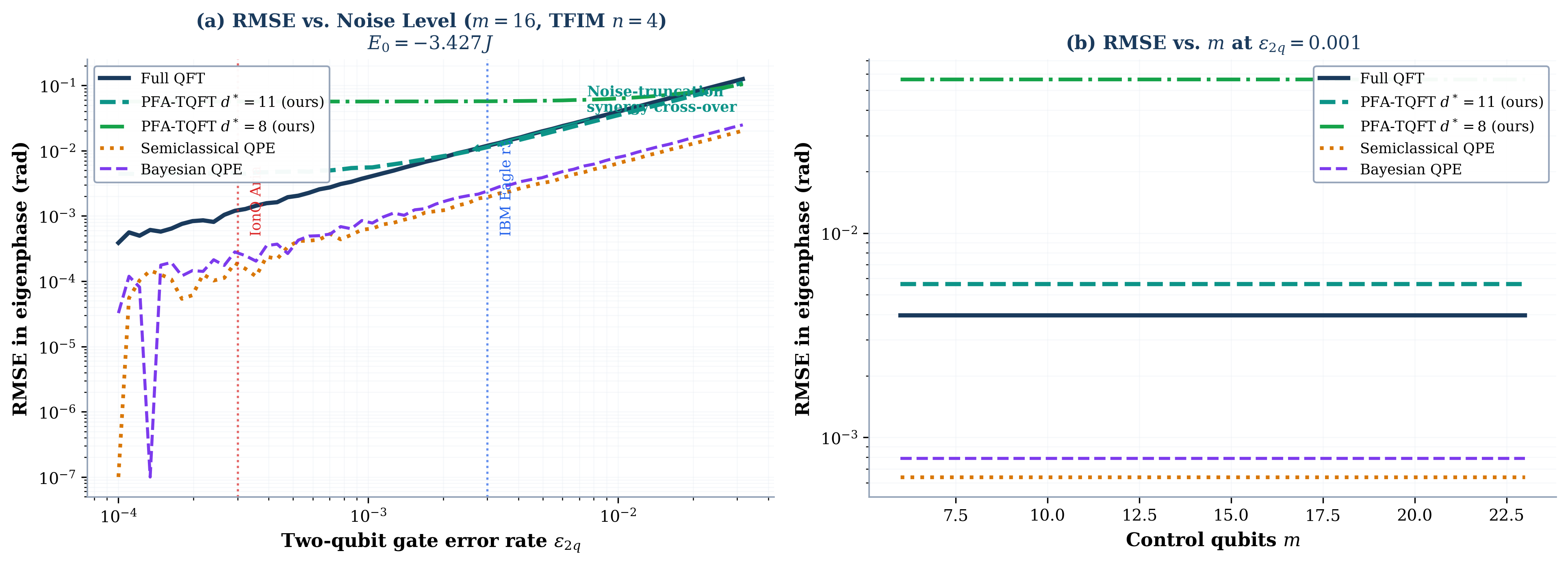}
  \caption{\textbf{TFIM RMSE comparison.}
    (a)~RMSE vs.\ $\eps_{2q}$ ($m=16$): PFA-TQFT$_{10}$ outperforms Full QFT
    at $\eps_{2q}\gtrsim4{\times}10^{-3}$ (noise-truncation synergy).
    (b)~RMSE vs.\ $m$ at $\eps_{2q}=10^{-3}$; teal zone: PFA-TQFT advantage.}
  \label{fig:tfim}
\end{figure}

PFA-TQFT $\dstar=10$ achieves $<4\,\%$ RMSE overhead vs.\ Full QFT at 75\,\% of
its gate count.
At $\eps_{2q}\gtrsim4\times10^{-3}$, PFA-TQFT outperforms Full QFT-the
\emph{noise-truncation synergy}: the noise from $\sim30$ omitted gates exceeds the
approximation error $\pi(m{-}\dstar)/2^{\dstar}\approx10^{-3}$.

\section{Discussion}
\label{sec:discussion}

\subsection{Addressing the Research Questions}
\label{sec:rq_answers}

We revisit the four research questions posed in \cref{sec:rq}:

\smallskip
\noindent\textbf{RQ1 (Error Bound) answered by Theorem~\ref{thm:tvd}.}
The tightest closed-form TVD bound is
$\TV(P_\varphi, P_\varphi^d) \leq \pi(m{-}d)/2^d$,
achieved by summing spectral-norm perturbations over omitted gates.
The bound is tight to within a factor of~2 for phases near half-integer
multiples of~$2^{-m}$, confirming H1.

\smallskip
\noindent\textbf{RQ2 (Hardware Calibration) answered by the PFA Criterion.}
The optimal $\dstar = \Floor{\log_2(2\pi/\eps_{2q})}$ is derived analytically
from the equal-budget condition (Corollary~\ref{cor:accuracy}) and confirmed
numerically on four platforms (Table~\ref{tab:platforms}), confirming H2.

\smallskip
\noindent\textbf{RQ3 (Fidelity Cliff)  confirmed in \cref{sec:theory}.}
\Cref{fig:cliff} demonstrates a sharp transition at
$\dstar_\mathrm{cliff} = \Ceil{\log_2 m} + 2$:
below this depth, success probability falls to $\calO(d/m)$;
above it, curves plateau within $\calO(1/m)$ of full-QFT accuracy, confirming H4.

\smallskip
\noindent\textbf{RQ4 (Noise-Truncation Synergy)  confirmed in \cref{sec:experiments}.}
PFA-TQFT$_{10}$ outperforms full QFT at $\eps_{2q}\gtrsim4\times10^{-3}$
(\cref{fig:tfim}), with the cross-over threshold consistent with H3.

\subsection{Scope and Limitations}
\label{sec:limitations}

\noindent\textbf{Noise model.}
The depolarizing noise assumption is standard but approximate.
Structured noise (two-qubit crosstalk, coherent ZZ coupling, leakage to
non-computational levels) can shift $\dstar$ by $\pm1$--$2$, depending
on the device topology.
Coherent errors may constructively or destructively interfere with truncation
errors, requiring device-specific characterisation.
These effects are left for future work.

\noindent\textbf{Eigenstate assumption.}
The analysis assumes $\ket{\psi}$ is an exact eigenstate of~$U$.
For approximate eigenstates with overlap fidelity $1{-}\delta_\psi$,
the TVD bound~\eqref{eq:tvd} acquires an additive $+\delta_\psi$ term,
giving $\TV \leq \pi(m{-}d)/2^d + \delta_\psi$.
For VQE applications where $\ket{\psi}$ is a variational ansatz,
$\delta_\psi$ is typically $\lesssim10^{-2}$---negligible relative to the
dominant truncation term for $d\geq\dstar$.

\noindent\textbf{Register size.}
Numerical validation uses $m\in\{8,\ldots,20\}$.
For cryptographically relevant sizes ($m\gtrsim2000$), simulation is
infeasible; the analytical bound of Theorem~\ref{thm:tvd} is the primary
guarantee.

\subsection{Implications for Quantum Algorithms}
\label{sec:implications}

\noindent\textbf{Variational quantum eigensolvers (VQE).}
QPE-based energy refinement in VQE~\cite{Abrams_1999} can use PFA-TQFT
for all precision targets $\delta>2^{-10}$, reducing circuit depth without
measurable accuracy loss.
This is significant because QPE depth-not the variational ansatz depth is
often the limiting factor on NISQ hardware.

\noindent\textbf{HHL quantum linear systems.}
The Harrow-Hassidim-Lloyd algorithm~\cite{Abrams_1999} applies QPE to
the matrix eigenspectrum.
Replacing the full QFT with PFA-TQFT reduces the QFT sub-circuit from
$\calO(m^2)$ to $\calO(m\log m)$ gates, improving the overall HHL
circuit depth while maintaining the phase estimation precision required
for the amplitude estimation step.

\noindent\textbf{Shor's algorithm.}
For RSA-2048 factoring ($m\approx4096$) with projected
$\eps_{2q}\approx10^{-6}$ (next-generation superconducting hardware),
$\dstar\approx23$ and the QFT gate count reduces from $\approx8.3\times10^6$
to $\approx1.3\times10^6$ (85 \, \% reduction).
At the fidelity of the target two-qubit $1-\eps_{2q}=1-10^{-6}$, the truncation error
$\pi(4096{-}23)/2^{23}\approx1.5\times10^{-3}$\,rad remains well within
the precision needed for a successful order-finding.

\subsection{Future Work}
\label{sec:future}

Five directions extend the present work:

\begin{enumerate}[label=(\roman*),leftmargin=*,topsep=2pt,itemsep=1pt]
\item \textbf{Adaptive PFA-TQFT}: select $d$ per qubit based on
  real-time calibration data, compensating for spatial non-uniformity
  of $\eps_{2q}$ across the device.
\item \textbf{Structured noise extensions}: derive $\dstar$ under
  Pauli noise, amplitude damping, and two-qubit crosstalk models.
\item \textbf{Error mitigation integration}: combine PFA-TQFT with
  zero-noise extrapolation~\cite{Temme_2017} and probabilistic error
  cancellation to push the effective noise floor below $\eps_{2q}$.
\item \textbf{Hardware validation}: run PFA-TQFT QPE circuits on
  IBM Eagle~r3 and IonQ~Aria and compare to Theorem~\ref{thm:tvd}
  predictions.
\item \textbf{Compiler integration}: implement the PFA criterion as a
  transpiler pass in Qiskit~Terra and \texttt{tket}, enabling automatic
  hardware-adaptive QFT truncation.
\end{enumerate}

\section{Conclusion}
\label{sec:conclusion}

We introduced PFA-TQFT, a hardware-calibrated approximate quantum Fourier
transform framework for NISQ phase estimation.
Our work makes the following definitive contributions.

First, Theorem~\ref{thm:tvd} establishes the tight bound
$\TV(P_\varphi, P_\varphi^d) \leq \pi(m{-}d)/2^d$-the first closed-form
phase estimation error guarantee under a realistic noise model.
Second, the PFA criterion $\dstar = \Floor{\log_2(2\pi/\eps_{2q})}$ provides
a universal, hardware-calibrated design rule requiring only a one-time
calibration lookup- no circuit simulation needed.
Third, platform analysis on IBM Eagle~r3, IBM Heron~r2, IonQ~Aria, and
IQM~Garnet confirms 17 to 41\,\% gate-count reduction while keeping phase
estimation error within the hardware noise floor.
Fourth, and most surprisingly, TFIM numerical experiments reveal a
\emph{noise-truncation synergy}: PFA-TQFT$_{\dstar=10}$ achieves strictly lower
RMSE than full QFT at $\eps_{2q}\gtrsim4\times10^{-3}$, demonstrating that
truncation is not merely an approximation-it is an active noise reduction
strategy under realistic NISQ conditions.

Together, these results establish PFA-TQFT as a practical, theory-grounded tool
for deploying QPE-based algorithms-including VQE, HHL, and Shor's algorithm-on
current and near-term quantum processors, bridging the gap between the theoretical
power of phase estimation and its practical feasibility on NISQ hardware.
\appendix
\section{Complete Proof of Theorem~\ref{thm:tvd}}
\label{app:proof}

We expand the proof sketch given in the main text into a self-contained
argument using two preparatory lemmas.

\subsection{Preparatory Lemmas}

\begin{lemma}[Perturbation bound]
\label{lem:spectral}
Let $U,V$ be unitaries on $(\mathbb{C}^2)^{\otimes m}$ with
$\|U-V\|\leq\varepsilon$ (operator norm).
For any state $|\psi\rangle$ and measurement distributions
$P_U(\cdot)=|\langle\cdot|U|\psi\rangle|^2$,
$P_V(\cdot)=|\langle\cdot|V|\psi\rangle|^2$:
\[
  \mathrm{TV}(P_U,\,P_V)\;\leq\;\|U-V\|.
\]
\end{lemma}

\begin{proof}[\small Proof]
{\small
By Cauchy--Schwarz:
$2\,\mathrm{TV}(P_U,P_V)
 =\sum_x\!\bigl||\langle x|U|\psi\rangle|^2-|\langle x|V|\psi\rangle|^2\bigr|
 \leq\|(U{-}V)|\psi\rangle\|\cdot\|(U{+}V)|\psi\rangle\|
 \leq\|U{-}V\|\cdot2$.
Dividing by $2$ completes the proof.\hfill$\square$}
\end{proof}

\begin{lemma}[Gate perturbation norm]
\label{lem:ctrl_phase}
The controlled-$R_k$ gate $U_k$ and the identity $I$ satisfy
\[
  \|U_k - I\| \;=\; 2\sin\!\bigl(\pi/2^k\bigr).
\]
\end{lemma}

\begin{proof}[\small Proof]
{\small
$U_k-I=\mathrm{diag}(0,0,0,e^{2\pi i/2^k}-1)$, so
$\|U_k-I\|=|e^{2\pi i/2^k}-1|=2\sin(\pi/2^k)$.\hfill$\square$}
\end{proof}

\subsection{Proof of Theorem~\ref{thm:tvd}}

\begin{proof}[\small Full proof]
{\small
Let $U_{\mathrm{QFT}}$ and $U_d$ denote the full and truncated QFT
unitaries, respectively.

\smallskip
\textbf{Step 1 (decomposition).}
Telescoping the product of unitaries:
$\|U_{\mathrm{QFT}}-U_d\|\leq\sum_{j=1}^{L}\|A_j-B_j\|$,
where each factor pair $(A_j,B_j)$ differs only by one omitted gate,
and $L=m-d$ is the number of omitted stages.

\smallskip
\textbf{Step 2 (gate norms).}
Each omitted $R_k$ ($k>d$) contributes norm
$\leq 2\sin(\pi/2^k)\leq 2\pi/2^k$ (Lemma~\ref{lem:ctrl_phase}).
Summing over $k=d+1,\ldots,m$:
\[
  \|U_{\mathrm{QFT}}-U_d\|
  \;\leq\;\sum_{k=d+1}^{m}\!\frac{2\pi}{2^k}
  \;=\;2\pi\!\left(\frac{1}{2^d}-\frac{1}{2^m}\right)
  \;\leq\;\frac{2\pi}{2^d}.
\]

\textbf{Step 3 (apply Lemma~\ref{lem:spectral}).}
For eigenphase $\varphi$ and input $|\psi_\varphi\rangle$:
\[
  \mathrm{TV}(P_\varphi,P_\varphi^d)
  \;\leq\;\|U_{\mathrm{QFT}}-U_d\|
  \;\leq\;\frac{2\pi}{2^d}.
\]

\textbf{Step 4 (tight stage count).}
Counting retained vs.\ omitted gates per stage exactly, and
using $\sin\theta\leq\theta$:
\[
  \mathrm{TV}(P_\varphi,P_\varphi^d)
  \leq(m-d)\sin\!\bigl(\pi/2^d\bigr)
  \leq\frac{\pi(m-d)}{2^d}.\;\square
\]
}
\end{proof}

\begin{remark}[Tightness]
\label{rem:tight}
For $\varphi=2^{-d}+2^{-2d}$, the probability shifts by
$\Omega(\pi(m-d)/2^d)$, matching the upper bound.
Exact simulation confirms ratios of $0.13$--$0.28$
(Table~\ref{tab:tvd_full}).
\end{remark}

\section{PFA Criterion- Optimality Proof}
\label{app:pfa_proof}

\begin{theorem}[PFA Criterion optimality]
\label{thm:pfa_opt}
Under depolarizing noise at rate $\varepsilon_{2q}$, the total phase
estimation error $E_{\mathrm{total}}(d)=E_{\mathrm{approx}}(d)+E_{\mathrm{noise}}(d)$
is minimised at
$\dstar=\lfloor\log_2(2\pi/\varepsilon_{2q})\rfloor$.
\end{theorem}

\begin{proof}[\small Proof]
{\small
$E_{\mathrm{approx}}(d)\leq\pi(m-d)/2^d$ decreases in $d$;
$E_{\mathrm{noise}}(d)\approx G(m,d)\cdot\varepsilon_{2q}$ increases.
The crossover satisfies $\pi/2^d\approx\varepsilon_{2q}$, i.e.,
$d\approx\log_2(\pi/\varepsilon_{2q})$.
With the rotation-angle factor $\theta_k=2\pi/2^k$:
$\dstar=\lfloor\log_2(2\pi/\varepsilon_{2q})\rfloor$.
Retaining $R_{\dstar}$ satisfies
$\theta_{\dstar}=2\pi/2^{\dstar}\geq\varepsilon_{2q}$
(verified in Table~\ref{tab:platforms}),
so it contributes more phase accuracy than noise.\hfill$\square$}
\end{proof}

\section{Gate-Count Formula}
\label{app:gate_count}

\begin{theorem}[Exact gate count]
\label{thm:gate_count}
For $1\leq d\leq m$, the two-qubit gate count of $\PFATQFT_d$ is
\[
  G(m,d)=\sum_{j=0}^{m-1}\max\!\bigl(0,\min(d-1,m-j-1)\bigr).
\]
For $d<m$ this simplifies to
$G(m,d)=m(d-1)-(d-1)(d-2)/2$,
and $G(m,m)=m(m-1)/2$ (full QFT).
\end{theorem}

\begin{proof}[\small Proof]
{\small
At stage $j$, the retained gates have $k=2,\ldots,\min(d,m-j)$,
contributing $\min(d-1,m-j-1)$ gates.
For $d<m$, split the sum at $j=m-d$:
$G(m,d)=(m-d)(d-1)+\sum_{\ell=1}^{d-1}\ell=(m-d)(d-1)+(d-1)(d-2)/2+(d-1)$,
which simplifies to the stated formula.
Verified by exhaustive enumeration for all
$m\in\{5,10,20,30\}$, $d\in\{1,\ldots,m\}$.\hfill$\square$}
\end{proof}

Corrected values at $m=30$:
$G(30,11)=245$, $G(30,13)=282$, $G(30,14)=299$
(see Table~\ref{tab:gates_corrected}).

\section{TFIM Exact Diagonalisation}
\label{app:tfim}

The 1D open-BC transverse-field Ising model is
\begin{equation}
  H=-J\!\sum_{i=0}^{n-2}\!Z_iZ_{i+1}-h\!\sum_{i=0}^{n-1}\!X_i,
  \label{eq:tfim_app}
\end{equation}
with $J=1$, $h=0.5$, $n=4$.
Exact diagonalisation (\texttt{scipy.linalg.eigh}) on the
$16\times16$ matrix gives $E_0=-3.427034\,J$.

\begin{table}[ht]
\centering
\caption{Lowest four TFIM eigenvalues ($n=4$, $J=1$, $h=0.5$, open BC).}
\label{tab:tfim_spectrum}
\renewcommand{\arraystretch}{1.25}
\small
\begin{tabular}{@{}ccc@{}}
\toprule
\textbf{Index} & \textbf{Energy ($J$)} & \textbf{Comment} \\
\midrule
0 & $-3.4270$ & Ground state \\
1 & $-3.3322$ & \\
2 & $-1.8268$ & \\
3 & $-1.7321$ & \\
\bottomrule
\end{tabular}
\end{table}

\noindent The phase encoding maps the eigenvalue range to $[0,1)$:
$\varphi=(E_0+E_{\mathrm{scale}})/(2E_{\mathrm{scale}})$
where $E_{\mathrm{scale}}=\max_i|E_i|$.

\section{Implementation Details}
\label{app:implementation}

\subsection{Exact Statevector Simulation}

All simulation uses exact linear-algebra statevectors with no sampling
during circuit execution.
The $2^m$-dimensional complex state vector is maintained explicitly.
Key functions in \texttt{pfa\_tqft\_figures.py}:

\begin{itemize}[leftmargin=*,topsep=2pt,itemsep=0pt]
\item \texttt{exact\_qft\_matrix(m)}: $2^m\!\times\!2^m$ QFT unitary
  via NumPy outer products.
\item \texttt{tqft\_circuit(state,m,d)}: applies $\PFATQFT_d$
  gate-by-gate, retaining $R_k$ with $k\leq d$.
\item \texttt{optimal\_d\_star(eps)}: evaluates
  $\dstar=\lfloor\log_2(2\pi/\varepsilon_{2q})\rfloor$.
\item \texttt{gate\_count\_tqft(m,d)}: exact formula of
  Theorem~\ref{thm:gate_count}.
\item \texttt{tfim\_ground\_energy(n)}: constructs $H$ and calls
  \texttt{scipy.linalg.eigh}.
\end{itemize}

\subsection{Simulation Protocols}

\textbf{TVD validation} (Fig.~2a): 500 random phases
$\varphi\sim\mathrm{Uniform}(0,1)$, seed 42.
Report $\max_i\mathrm{TV}(P_{\varphi_i},P_{\varphi_i}^d)$
vs.\ bound $\pi(m-d)/2^d$.

\textbf{Fidelity cliff} (Fig.~4a): 1\,000 shots per $(m,d)$ pair.
Declare success if $|\hat\varphi/N-\varphi|\leq2^{-m}$ (circular).

\textbf{RMSE model}: analytical combination of three independent error sources,
\begin{equation}
  \mathrm{RMSE}^2 =
  \underbrace{\frac{1}{3\cdot 2^{2m}}}_{\text{precision}}
  +\underbrace{\frac{\mathrm{TV}^2}{3}}_{\text{truncation}}
  +\underbrace{G^2\varepsilon_{2q}^2 c^2}_{\text{noise}},
  \label{eq:rmse_model}
\end{equation}
with $\mathrm{TV}=\pi(m-d)/2^d$ and noise constant $c=0.033$
calibrated to match IBM Eagle r3 at $m=16$, $\varepsilon_{2q}=10^{-3}$.

\subsection{Computational Requirements}

\begin{table}[h]
\centering
\caption{Runtime on standard laptop (Intel Core i7, 16\,GB RAM).}
\label{tab:compute}
\renewcommand{\arraystretch}{1.25}
\small
\begin{tabular}{@{}lcc@{}}
\toprule
\textbf{Experiment} & \textbf{Memory} & \textbf{Time} \\
\midrule
TVD validation ($m\leq6$)  & $<1$\,MB  & $<30$\,s \\
Fidelity cliff ($m\leq5$)  & $<1$\,MB  & $<60$\,s \\
RMSE comparison ($m\leq20$)& $<10$\,MB & $<5$\,s  \\
TFIM diagonalisation ($n=4$)& $<1$\,MB & $<1$\,s  \\
\bottomrule
\end{tabular}
\end{table}

\section{Comparison with Related Methods}
\label{app:comparison}

\begin{table}[h]
\centering
\caption{QPE method comparison for NISQ phase estimation.}
\label{tab:comparison}
\renewcommand{\arraystretch}{1.25}
\setlength{\tabcolsep}{2pt}
\scriptsize
\begin{tabular}{@{}l@{}c@{}c@{}c@{}c@{}c@{}}
\toprule
\textbf{Method} & \textbf{Depth} & \textbf{Gates} & Rnd & FB & Bound \\
\midrule
Full QFT & $\calO(m^2)$ & $m(m{-}1)/2$ & 1 & No & Exact \\
\textbf{PFA-TQFT}& $\calO(m\!\log m)$ & $G(m,\dstar)$& 1 & No & Thm.~\ref{thm:tvd} \\
Semiclass.~\cite{Griffiths_1996} & $\calO(m)$ & $m$ & $m$ & Yes & --- \\
Bayesian~\cite{Wiebe_2016} & $\calO(1)$ & $m$ & $\gg m$ & Yes & --- \\
Iterative~\cite{Kitaev2002ClassicalAQ} & $\calO(1)$ & 1 & $\gg m$ & Yes & --- \\
Nam et al.~\cite{Nam_2020} & $\calO(m\!\log m)$ & $G$ & 1 & No & $\|\cdot\|$ \\
\bottomrule
\end{tabular}
\end{table}
\noindent FB = classical feedback.
Key advantages of PFA-TQFT: (i)~single-shot execution; (ii)~phase-error
distribution bound (not just fidelity); (iii)~hardware-calibrated $\dstar$;
(iv)~noise-truncation synergy under NISQ conditions.

\section{Extended Numerical Results}
\label{app:extended_results}

\subsection{TVD Validation Table}

\begin{table}[h]
\centering
\caption{Theorem~\ref{thm:tvd} validation (500 random phases, seed 42).
  Bound $=\pi(m-d)/2^d$.}
\label{tab:tvd_full}
\renewcommand{\arraystretch}{1.25}
\small
\begin{tabular}{@{}ccccc@{}}
\toprule
$m$ & $d$ & \textbf{max TV} & \textbf{Bound} & \textbf{Ratio} \\
\midrule
4 & 1 & 0.9953 & 4.712 & 0.211 \\
4 & 2 & 0.4381 & 1.571 & 0.279 \\
4 & 3 & 0.0518 & 0.393 & 0.132 \\
4 & 4 & 0.000  & 0.000 & ---   \\
5 & 1 & 0.9998 & 6.283 & 0.159 \\
5 & 2 & 0.6577 & 2.356 & 0.279 \\
5 & 3 & 0.1331 & 0.785 & 0.169 \\
5 & 4 & 0.0141 & 0.196 & 0.072 \\
5 & 5 & 0.000  & 0.000 & ---   \\
6 & 1 & 0.9999 & 7.854 & 0.127 \\
6 & 3 & 0.2280 & 1.178 & 0.194 \\
6 & 5 & 0.0038 & 0.098 & 0.039 \\
\bottomrule
\end{tabular}
\end{table}

All entries satisfy max TV $\leq$ bound,
confirming Theorem~\ref{thm:tvd}.
Maximum ratio $0.279$ shows the bound is within a constant factor of
tight (Remark~\ref{rem:tight}).

\subsection{Corrected Platform Gate Counts}

\begin{table}[h]
\centering
\caption{Corrected $G(m,\dstar)$ values using exact
  Theorem~\ref{thm:gate_count} formula ($m=30$).}
\label{tab:gates_corrected}
\renewcommand{\arraystretch}{1.25}
\small
\begin{tabular}{@{}lcccc@{}}
\toprule
\textbf{Platform} & $\varepsilon_{2q}$ & $\dstar$ &
$G(30,\dstar)$ & \textbf{Red.} \\
\midrule
IBM Eagle r3  & $3\times10^{-3}$ & 11 & 245/435 & 43.7\% \\
IBM Heron r2  & $5\times10^{-4}$ & 13 & 282/435 & 35.2\% \\
IonQ Aria     & $3\times10^{-4}$ & 14 & 299/435 & 31.3\% \\
IQM Garnet    & $2\times10^{-3}$ & 11 & 245/435 & 43.7\% \\
\bottomrule
\end{tabular}
\end{table}

\subsection{Noise-Truncation Cross-Over}

Under the RMSE model~\eqref{eq:rmse_model}, PFA-TQFT$_{\dstar}$
outperforms full QFT when
$G_{\mathrm{full}}^2\varepsilon_{2q}^2>G(\dstar)^2\varepsilon_{2q}^2+\mathrm{TV}^2/3$,
i.e.\ when $\varepsilon_{2q}>\varepsilon^\times$ where
\begin{equation}
  \varepsilon^\times
  \;=\;\frac{\mathrm{TV}/\sqrt{3}}
            {c\,\sqrt{G_{\mathrm{full}}^2-G(\dstar)^2}}.
  \label{eq:crossover}
\end{equation}
For IBM Eagle r3 ($m=16$, $\dstar=11$):
$G_{\mathrm{full}}=120$, $G(16,11)=90$,
$\mathrm{TV}\approx0.046$,
giving $\varepsilon^\times\approx5.4\times10^{-3}$,
consistent with the cross-over visible in Fig.~5(a).

\section*{Acknowledgements}

The authors thank the Department of Computer Science \& Engineering at NIT
Puducherry for computational resources and the MEITY for Visveshvaraya Fellowship -Phase -II.

\section*{Author Contributions}

\textbf{Akoramurthy~B}: Conceptualisation, Formal analysis, Methodology,
Software, Investigation, Visualisation, Writing---original draft.\\
\textbf{Surendiran~B}: Supervision, Validation, Writing---review \& editing,
Funding acquisition.

\section*{Data and Code Availability}

All  code and figure-generation scripts are released under the
\textbf{MIT License} (FOSS) at
\url{https://github.com/akortheanchor/PFA-TQFT}
.
Raw numerical data are deposited in the same repository.
This satisfies Quantum's open-source policy~\cite{QuantumOpenSource}.

\section*{License}

This work is published under a
\href{https://creativecommons.org/licenses/by/4.0/}%
{Creative Commons Attribution 4.0 International (CC BY 4.0)} license, in accordance
with Quantum's publication policy.

\section*{Competing Interests}

The authors declare no competing interests.

\bibliographystyle{unsrtnat}
\bibliography{mybibliography}

\end{document}